\documentclass[prl,superscriptaddress,twocolumn]{revtex4}
\usepackage{amsfonts}
\usepackage{amsmath}
\usepackage{amsthm}
\usepackage{amscd}
\usepackage{amssymb}
\usepackage{subfigure}
\usepackage{amsxtra}
\usepackage{bm}           
\usepackage{bbm}
\usepackage{graphicx}
\usepackage{epstopdf}
\usepackage{color}
\usepackage{enumitem}
\usepackage{tabularx}
\usepackage[normalem]{ulem}

\usepackage[utf8]{inputenc}
\usepackage{physics}
\usepackage{framed}
\usepackage{graphicx}
\usepackage{float}

\usepackage{algorithm}
\floatname{algorithm}{Protocol}
\usepackage{algorithmic}

\newcommand{\be}{\begin{equation}}
\newcommand{\ee}{\end{equation}}
\newcommand{\ba}{\begin{array}}
\newcommand{\ea}{\end{array}}
\newcommand{\bqa}{\begin{eqnarray}}
\newcommand{\eqa}{\end{eqnarray}}

\newtheorem{theorem}{Theorem}
\newtheorem{lemma}[theorem]{Lemma}
\newtheorem{corollary}[theorem]{Corollary}
\newenvironment{proof-sketch}{%
  \proof}{\endproof}

\newtheorem{manualtheoreminner}{Theorem}
\newenvironment{manualtheorem}[1]{%
  \renewcommand\themanualtheoreminner{#1}%
  \manualtheoreminner
}{\endmanualtheoreminner}

\newcounter{protocol}

\begin{document}

\title{Anonymity for practical quantum networks}

\author{Anupama Unnikrishnan}\affiliation{Department of Atomic and Laser Physics, Clarendon Laboratory, University of Oxford, Oxford OX1 3PU, UK}
\author{Ian J. MacFarlane}\affiliation{Massachusetts Institute of Technology, Cambridge, Massachusetts, USA}
\author{Richard Yi}\affiliation{Massachusetts Institute of Technology, Cambridge, Massachusetts, USA}
\author{Eleni Diamanti}\affiliation{LIP6, CNRS, Sorbonne Universit\'e, 75005 Paris, France}
\author{Damian Markham}\affiliation{LIP6, CNRS, Sorbonne Universit\'e, 75005 Paris, France}
\author{Iordanis Kerenidis}\affiliation{IRIF, CNRS, Universit\'e Paris Diderot, Sorbonne Paris Cit\'e, 75013 Paris, France}

\begin{abstract}
Quantum communication networks have the potential to revolutionise information and communication technologies. Here we are interested in a fundamental property and formidable challenge for any communication network, that of guaranteeing the anonymity of a sender and a receiver when a message is transmitted through the network, even in the presence of malicious parties. We provide the first practical protocol for anonymous communication in realistic quantum networks.
\end{abstract}

\date{\today}

\maketitle
The rapid development of quantum communication networks will allow a large number of agents with different technological, classical or quantum, capabilities to securely exchange messages and perform efficiently distributed computational tasks, opening new perspectives for information and communication technologies and eventually leading to the quantum internet \cite{Kim:nature08}. Many applications of quantum networks are known, including, for example, quantum key distribution (QKD) \cite{SBC:rmp09,DLQ:npjqi16} or blind and verifiable delegation of quantum computation \cite{GKK:tcs18}, and many more are yet to be developed.

A crucial yet challenging functionality required in any network is the ability to guarantee the anonymity of two parties, the Sender and the Receiver, when they wish to transmit a message through the network. In a realistic network, anonymity should be guaranteed in the presence of malicious parties. We would additionally like that this happens in an information-theoretic setting, meaning without making any assumptions neither on the number nor on the computational power of these malicious parties, who might in fact have a quantum computer in their hands.

In the classical setting, anonymity, as well as any multiparty secure computation, is possible with information-theoretic security when there is an honest majority of agents.  Furthermore, Broadbent and Tapp \cite{BT:asiacrypt07} showed how to anonymously transmit a classical message, as well as a number of other secure protocols, in the absence of an honest majority.
In order to do this, secure pairwise classical channels are required, as well as classical broadcast channels.

In the quantum setting, the first work to deal with the anonymity of quantum messages was that of Christandl and Wehner \cite{CW:asiacrypt05}. In their work, one assumes that the $n$ agents share a perfect $n$-party GHZ state, \emph{i.e.}, the state $\frac{1}{\sqrt{2}} ( \ket{0^n}+\ket{1^n} )$ \cite{GHZ}. Under this assumption, they provide protocols with perfect anonymity both for the broadcast of a classical bit and for the creation of an EPR pair between a Sender and a Receiver. Then, they combine the two protocols in order to transmit a quantum message using a teleportation scheme \cite{BBC:prl93}. This first creates an EPR pair anonymously between Sender and Receiver, and then the Sender transmits the two classical outcomes of her measurements anonymously.
The advantage of this protocol is that it only involves local operations and classical communication (LOCC) once the GHZ state is shared between the agents. However, it requires the assumption that a perfect GHZ state has been honestly shared between the agents. More recently, Lipinska \emph{et al.} \cite{LMW:pra18} showed how to perform a similar protocol starting from trusted W states, albeit only probabilistically.

In order to remedy the drawback of a perfect shared quantum state, Brassard \emph{et al.} \cite{BBF:asiacrypt07} devised a different protocol, which 
includes a verification stage for ensuring that the shared state is at least symmetric with respect to the honest agents, and hence perfect anonymity is preserved. This test involves each agent performing a controlled-NOT operation between her initial quantum bit (qubit) and $n-1$ fresh ancilla qubits that she then sends to all other agents. Each agent then measures $n-1$ qubits in the subspace spanned by the all zeros and all ones strings and if the measurement accepts then the protocol continues with the remaining $n$-party GHZ state. While the authors manage in this way to preserve perfect anonymity, their protocol cannot be easily implemented, since each agent needs to perform a size-$n$ quantum circuit and also to have access to quantum communication with all other agents.

We address this problem by considering quantum anonymous transmission in the presence of an untrusted source that may not be producing the GHZ state. Our two main ingredients are the Christandl-Wehner protocol for anonymous entanglement \cite{CW:asiacrypt05}, and a protocol for verifying GHZ states described in Ref. \cite{PAW:prl12}. We then present a new notion of approximate anonymity that is appropriate for realistic quantum networks, and show a practical and efficient protocol to achieve such anonymity in the transfer of a quantum message.

\textbf{Communication scenario}.---Let us first describe the communication scenario we consider. Our network consists of $n$ agents who can perform local operations and measurements. A source, who may be malicious, produces GHZ states that our agents wish to use for anonymous quantum communication. The source may produce a different state in every round, or even entangle the states between different rounds. 

The agents themselves may be honest or malicious. Honest agents follow the protocol but malicious agents can collaborate with the source, work together, and apply any cheating strategy on their systems, including entangling them with some ancilla that they may store in memory to be accessed at will. The aim of the malicious agents is to break the anonymity or security of the protocol.

In addition to public quantum channels between all agents, we require some classical communication channels. More specifically, we assume there are private classical channels between each pair of agents. This can be ensured by each pair of agents sharing a private random string, and is a standard assumption if we have malicious agents in a classical network. Furthermore, each agent has access to a broadcast channel, which she can use to send classical information to all other agents. We will use the term simultaneous broadcast when it is required that all agents must broadcast their bit simultaneously, which is an impractical resource as it is hard to ensure in practice. Crucially, we only need a regular (or non-simultaneous) broadcast channel in our anonymous quantum communication protocol; all the subprotocols that we use remove the requirement of simultaneous broadcasting.

\textbf{Anonymous classical protocols}.---We start by providing the details of a few known anonymous classical protocols, some of which we will use directly. First, there exists a classical private protocol from Ref. \cite{BT:asiacrypt07}, \textsf{LogicalOR}, where each agent inputs a single bit and the protocol computes the logical OR of these bits.
This protocol has correctness in that if the input of all agents is $0$, the protocol always outputs the correct answer  (\emph{i.e.} $0$). If any agent inputs $1$, this protocol succeeds (\emph{i.e.} outputs $1$) with probability $1 - 2^{-S}$ after $S$ rounds.
Privacy here means that only the agent can know their input. \textsf{LogicalOR} is built using another protocol \textsf{Parity} \cite{BT:asiacrypt07}, which privately computes the parity of the input string; however, contrary to the the \textsf{Parity} protocol, \textsf{LogicalOR} does not require a simultaneous broadcast channel. Further details of both protocols are given in Appendix A.

We will use the \textsf{LogicalOR} protocol in order to create the functionality \textsf{RandomBit}, given in Protocol 1, which allows the Sender to anonymously choose a random bit according to some probability distribution $D$. The correctness and privacy of \textsf{RandomBit} follow directly from the properties of \textsf{LogicalOR}, namely the only thing the malicious agents learn is the bit chosen by the Sender, but not who the Sender is. We then extend the \textsf{RandomBit} functionality to define a \textsf{RandomAgent} functionality, where the Sender privately picks a random agent by performing the \textsf{RandomBit} protocol $\log_2 n$ times.

\begin{algorithm}[H]
\caption{\textsf{RandomBit}}
\begin{flushleft}
\textit{Input:} All: parameter $S$. Sender: distribution $D$. \\
\textit{Goal:} Sender chooses a bit according to $D$.
\end{flushleft}
\begin{algorithmic}[1]
\STATE The agents pick bits $\{ x_i \}_{i = 1}^n$ as follows: the Sender picks bit $x_i$ to be 0 or 1 according to distribution $D$; all other agents pick $x_i =0$. \\ \
\STATE Perform the \textsf{LogicalOR} protocol with input $\{ x_i \}_{i=1}^n$ and security parameter $S$ and output its outcome.
\end{algorithmic}
\end{algorithm}

Last, we need the \textsf{Notification} functionality \cite{BT:asiacrypt07}, given in Protocol 2, where the Sender anonymously notifies an agent as the Receiver.
Note that we use the same security parameter $S$ throughout for simplicity, however this is not required. As we explicitly call on this in our main protocol, we describe it below.

\begin{algorithm}[H]
\caption{\textsf{Notification} \cite{BT:asiacrypt07}}
\begin{flushleft}
\textit{Input}: Security parameter $S$, Sender's choice of Receiver is agent $r$. \\
\textit{Goal}: Sender notifies Receiver.
\end{flushleft}
\begin{algorithmic}[1]
\STATE For each agent $i$:
\begin{enumerate}
\item[(a)] Each agent $j \neq i$ picks $p_j$ as follows: if $i = r$ and agent $j$ is the Sender, then $p_j = 1$ with probability $\frac{1}{2}$ and $p_j = 0$ with probability $\frac{1}{2}$. Otherwise, $p_j = 0$. Let $p_i = 0$.
\item[(b)] Run the \textsf{Parity} protocol with input $\{p_i\}_{i=1}^n$, with the following differences: agent $i$ does not broadcast her value, and they use a regular broadcast channel rather than simultaneous broadcast. If the result is $1$, then $y_i = 1$.
\item[(c)] Repeat steps 1(a) - (b) $S$ times. If the result of the \textsf{Parity} protocol is never 1, then $y_i = 0$. \\ \
\end{enumerate}
\STATE If agent $i$ obtained $y_i = 1$, then she is the Receiver.
\end{algorithmic}
\end{algorithm}

\textbf{Anonymous entanglement with perfect trusted GHZ states}.---In addition to the previous classical protocols, we will need the \textsf{Anonymous Entanglement} protocol from Ref. \cite{CW:asiacrypt05}, given in Protocol 3. Here, it is assumed that the agents share a state which in the honest case is the GHZ state, and that the Sender and Receiver know their respective identities. It is not hard to see that assuming the initial state is a perfect GHZ state, then the protocol creates an EPR pair between the Sender and the Receiver perfectly anonymously.

\begin{algorithm}[H]
\caption{\textsf{Anonymous Entanglement} \cite{CW:asiacrypt05}}
\begin{flushleft}
\textit{Input}: $n$ agents share a GHZ state. \\
\textit{Goal}: EPR pair shared between Sender and Receiver.
\end{flushleft}
\begin{algorithmic}[1]
\STATE Each agent, apart from the Sender and Receiver, applies a Hadamard transform to their qubit. They measure in the computational basis and broadcast their outcome. \\ \
\STATE The Sender first picks a random bit $b$, broadcasts it, and applies a phase flip $\sigma_z$ only when $b=1$. \\ \
\STATE The Receiver picks a random bit $b'$, broadcasts it and applies a phase flip $\sigma_z$ only when the parity of everyone else's broadcasted bits is 1.
\end{algorithmic}
\end{algorithm}

\textbf{Efficient verification of GHZ states}.---The last ingredient we use is the \textsf{Verification} protocol for GHZ states from the work of Pappa \emph{et al.} \cite{PAW:prl12} that was also implemented for 3- and 4-party GHZ states in McCutcheon \emph{et al.} \cite{MPB:natcomm16}.
There, one of the agents, the Verifier, would like to verify how close the shared state is to the ideal state. Let $k$ be the number of honest agents. The verification protocol is then given in Protocol 4.

\begin{algorithm}[H]
\caption{\textsf{Verification} \cite{PAW:prl12,MPB:natcomm16}}
\begin{flushleft}
\textit{Input}: $n$ agents share state $\ket{\Psi}$. \\
\textit{Goal}: GHZ verification of $\ket{\Psi}$ for $k$ honest agents.
\end{flushleft}
\begin{algorithmic}[1]
\STATE The Verifier generates random angles $\theta_j \in [0,\pi)$ for all agents including themselves ($j\in[n]$), such that $\sum_j \theta_j$ is a multiple of $\pi$. The angles are then sent out to all the agents in the network. \\ \
\STATE Agent $j$ measures in the basis $\{\ket{+_{\theta_j}},\ket{-_{\theta_j}}\}=\{\frac{1}{\sqrt{2}}(\ket{0}+e^{i\theta_j}\ket{1}),\frac{1}{\sqrt{2}}(\ket{0}-e^{i\theta_j}\ket{1})\}$, and sends the outcome $Y_j=\{0,1\}$ to the Verifier. \\ \
\STATE The state passes the verification test when the following condition is satisfied: if the sum of the randomly chosen angles is an even multiple of $\pi$, there must be an even number of $1$ outcomes for $Y_j$, and if the sum is an odd multiple of $\pi$, there must be an odd number of $1$ outcomes for $Y_j$. We can write this condition as
$
\bigoplus_j Y_j=\frac{1}{\pi}\sum_j\theta_j\pmod 2.
$
\end{algorithmic}
\end{algorithm}

From the proofs in Refs.\@ \cite{PAW:prl12} and \cite{MPB:natcomm16}, one can see that the ideal state always passes the verification test, and, more interestingly, a soundness statement can also be proven. As in \cite{PAW:prl12}, we take the ideal $n$-party state to be $\ket{\Phi_0^n}$, given by:
\begin{align}
\ket{\Phi_0^n} = \frac{1}{\sqrt{2^{n-1}}} \Big[ \underset{\Delta(y) = 0 \text{ (mod 4)}}{\sum} \ket{y} - \underset{\Delta(y) = 2 \text{ (mod 4)}}{\sum} \ket{y} \Big],\nonumber
\end{align}
where $\Delta(y) = \sum_i y_i$ denotes the Hamming weight of the classical $n$-bit string $y$. This state is equivalent to the GHZ state up to local unitaries.
Analogous to \cite{PAW:prl12, MPB:natcomm16}, to measure the quality of the state $\ket{\Psi}$ shared between the $n$ agents, we take a fidelity measure given by $F'(\ket{\Psi}) = \underset{U}{\max \ } F(U \ket{\Psi}, \ket{\Phi_0^n})$, where $U$ is any unitary operation on the space of the malicious agents.
This reflects the fact that we are concerned with certifying the state up to operations on the malicious parts, since these are in any case out of the control of the honest agents.
Then, even assuming the malicious agents apply their optimal cheating strategy, the probability of passing the test with the state $\ket{\Psi}$, denoted by $P(\ket{\Psi})$, satisfies $F'(\ket{\Psi}) \geq 4 P(\ket{\Psi}) - 3$ \cite{PAW:prl12, MPB:natcomm16}. Note that this holds even if the shared state is mixed; however, as we will see later, a clever malicious source will always create pure states.

For our purposes, we will use below a version of this verification protocol that is similar to the Symmetric Verification protocol in Ref. \cite{PAW:prl12}. There, it was shown that with the use of a trusted common random string it was possible for all agents to take random turns verifying the validity of the GHZ state. This leads to the guarantee that if the state is accepted a large number of times before the agents decide to use it, then with high probability, when the state is used it should be very close to the correct one.

\textbf{Anonymity for realistic quantum networks}.---All the quantum protocols we have seen that are used to achieve anonymity assume perfect operations and achieve perfect anonymity. In practice, of course, no operation can be perfect and hence perfect anonymity is unattainable. Nevertheless, it is still possible to define an appropriate notion of anonymity that is relevant for practical protocols.

We define the notion of an {\em $\epsilon-$anonymous} protocol, where for any number $n-k$ of malicious agents out of $n$ agents in total, the malicious agents, even when they have in their possession the entire quantum state that corresponds to the protocol, can only guess who the Sender is (even when the Receiver is malicious) or who the Receiver is, with probability that is bounded by $\frac{1}{k}+\epsilon$. The perfect anonymity is defined when $\epsilon$ is equal to 0.

\textbf{Efficient anonymous quantum message transmission}.---We will now show how to devise an efficient $\epsilon$-anonymous protocol for quantum message transmission.
For simplicity we assume there is only one Sender. If not, the agents can run a simple classical protocol in the beginning of the protocol in order to deal with collisions (multiple Senders) and achieve the unique Sender property. See also Refs. \cite{BT:asiacrypt07} and \cite{CW:asiacrypt05} for details.

Moreover, for simplicity we will describe a protocol where we distribute one EPR pair between the Sender and the Receiver. Then one can perform anonymous teleportation of the classical measurement results, using in particular the \textsf{Fixed Role Anonymous Message Transmission} functionality as was described in Ref. \cite{BT:asiacrypt07}. In case we want to increase the fidelity of the transmitted quantum message, we can further use the subroutines from Brassard \emph{et al.} \cite{BBF:asiacrypt07} which first create a number of non-perfect EPR pairs, then distill one pair and then perform the teleportation. Given that our main contribution is the efficient anonymous protocol for the GHZ verification, we do not provide here these details that are explained in Ref. \cite{BT:asiacrypt07}.

Our scheme is outlined in Protocol 5.

\begin{algorithm}[H]
\caption{\textsf{$\epsilon$-Anonymous Entanglement Distribution}}
\begin{flushleft}
\textit{Input}: Security parameter $S$.\\
\textit{Goal}: EPR pair created between Sender and Receiver with $\epsilon$-anonymity.
\end{flushleft}
\begin{algorithmic}[1]
    \STATE {\bf The Sender notifies the Receiver:}

    The agents run the \textsf{Notification} protocol. \\ \ 

    \STATE {\bf GHZ state generation:}

    The source generates a state $\ket{\Psi}$ and distributes it to the agents. \\ \

    \STATE {\bf The Sender anonymously chooses Verification or Anonymous Entanglement:}
    \begin{enumerate}
    \item[(a)] The agents perform the \textsf{RandomBit} protocol, with the Sender choosing her input according to the following probability distribution: she flips $S$ fair classical coins, and if all coins are heads, she inputs $0$, else she inputs $1$. Let the outcome be $x$.
    \item[(b)] If $x=0$, the agents run \textsf{Anonymous Entanglement}, \;\;\;
    else if $x=1$:
    \begin{enumerate}
    \item[(i)] Run the \textsf{RandomAgent} protocol, where the Sender inputs a uniformly random $j \in [n]$, to get output $j$.
    \item[(ii)] Agent $j$ runs the \textsf{Verification} protocol as the Verifier, and if she accepts the outcome of the test they return to step 2, otherwise the protocol aborts.
    \end{enumerate}
    \end{enumerate}
If at any point in the protocol, the Sender realises someone does not follow the protocol, she stops behaving like the Sender and behaves as any agent.
\end{algorithmic}
\end{algorithm}

We are now ready to analyse the above protocol. First, note that if the state is a perfect GHZ state and the operations of the honest agents are perfect, then the anonymity of the protocol is perfect.

In step 1, the agents run the \textsf{Notification} protocol which is perfectly anonymous. In the second step, the GHZ state is shared between the agents, which does not affect the anonymity. Note that the role of the source can be played by an agent, as long as the choice of the agent is independent of who the Sender is. In step 3(a), the agents run the \textsf{RandomBit} protocol which is also perfectly anonymous. The analysis of the step 3(b) follows from the analysis of the Symmetric Verification protocol in Ref. \cite{PAW:prl12}. The only difference here is that instead of using a common random string, it is the Sender who picks the randomness uniformly. Thus, since the input of the Sender completely determines the outcome of the protocol, the Sender can immediately see if her choice does not correspond to the outcome, and hence only continues if the randomness is perfectly uniform.

Let $C_\epsilon$ be the event that the above protocol does not abort and that the state used for the \textsf{Anonymous Entanglement} protocol is such that no matter what operation the malicious agents do to their part, the fidelity of the state with the GHZ state is at most $\sqrt{1 - \epsilon^2}$. Then, we prove the following Theorem for the honest agents:
\begin{theorem}
For all $\epsilon >0$,
\begin{equation}
\text{Pr}[C_\epsilon] \leq  2^{-S}\frac{4n}{1 - \sqrt{1 - \epsilon^2}}.
\end{equation}
\end{theorem}
\begin{proof-sketch}
As proved in \cite{PAW:prl12}, the optimal cheating strategy of a malicious source, which maximises the probability of $C_\epsilon$, is to create in each round of the protocol a pure state $\ket{\Psi}$ such that $F'(\ket{\Psi}) = \sqrt{1-\epsilon^2}$.

The probability of event $C_\epsilon$ is then given by the probability of the state being used and all the tests being passed in the previous rounds. This in turn will depend on the success probability of \textsf{RandomBit}, and if the agent chosen to act as the Verifier is honest. Given that a state with $F'(\ket{\Psi})$ passes the verification protocol with probability $P(\ket{\Psi})$, we can then determine a bound on $\text{Pr}[C_\epsilon]$ by following the proof in Ref. \cite{PAW:prl12}. The full proof is given in Appendix B.

\end{proof-sketch}
By taking $S= \log_2 (\frac{4n}{(1- \sqrt{1 - \epsilon^2})\delta})$, we have $\text{Pr}[C_\epsilon] \leq \delta$.
Let us assume for simplicity that when the event $C_\epsilon$ is true, which happens with probability at most $\delta$, the malicious agents can perfectly guess the Sender or the Receiver. We will now see that when the event $C_\epsilon$ is false, which happens with probability at least $1-\delta$, the malicious agents cannot guess the Sender or the Receiver with probability much higher than a random guess.
In other words, there is no strategy for breaking the anonymity of the communication that works much better than simply guessing an honest agent at random.

Note that $C_\epsilon$ being false means that the fidelity of the shared state with the GHZ state (up to a local operation on the malicious agents) is at least $\sqrt{1-\epsilon^2}$.
By doing enough rounds, we can ensure that the probability of $C_\epsilon$ is negligible.
Our statement of anonymity is given as follows:
\begin{theorem}
If the agents share a state $\ket{\Psi}$ such that $F'(\ket{\Psi}) \geq \sqrt{1-\epsilon^2}$, then the probability that the malicious agents can guess the identity of the Sender is given by:
\begin{align}
\text{Pr}[\text{guess}] & \leq \frac{1}{k} + \epsilon.
\end{align}
\label{th:anon}
\end{theorem}
\begin{proof-sketch}
First, we show that when the shared state is close to the GHZ state (up to some operation $U$ on the malicious agents' part of the state), then the fidelity between the final state of the protocol when the Sender is agent $i$, $\ket{\Psi_i}$, and the final state of the protocol when the Sender is agent $j$, $\ket{\Psi_j}$, is high.

Then, we show that when the fidelity between the states $\ket{\Psi_i}$ and $\ket{\Psi_j}$ is close to 1, the probability that the malicious agents can guess the identity of the Sender is close to a random guess. The full proof is given in Appendix C.
\end{proof-sketch}

Finally, we consider the entangled state created anonymously between the Sender and Receiver. Although we have not considered a particular noise model, our analysis incorporates a reduced fidelity of $\ket{\Psi}$, the state shared by all the agents at the beginning of the protocol. We can carry this forward to the resulting anonymously entangled state, if we assume all the agents are honest and have followed the protocol. We find that the fidelity of the final entangled state with the EPR pair will be at least the fidelity of $\ket{\Psi}$ with the GHZ state. After the entangled state has been constructed, the Sender and Receiver can perform anonymous teleportation of any quantum message $\ket{\phi}$ by anonymously sending a classical message with the teleportation results. Our final statement is then given in Corollary \ref{cor:anon}.

\begin{corollary}
Using Protocol 5, we can achieve an $\epsilon$-anonymous protocol for quantum message transmission.
\label{cor:anon}
\end{corollary}

\textbf{Discussion}.---We have proposed a practical protocol for anonymous quantum communications in the presence of malicious parties and an untrusted source. The verification step is carried out using a protocol that has been experimentally demonstrated \cite{MPB:natcomm16}, and is tolerant to losses and noise by design. Our protocol achieves in this full adversarial scenario an approximate notion of anonymity that we call $\epsilon$-anonymity and which is relevant in the context of realistic quantum networks.

While the scheme in Ref. \cite{BBF:asiacrypt07} results in an exponential scaling, their protocol is not easily implementable.
Recent work in Ref. \cite{LMW:pra18} provides a protocol for anonymous transmission using the W state rather than the GHZ state. While this is beneficial in terms of robustness to noise, the protocol creates the anonymously entangled state only with a probability $2/n$. Furthermore, the security analysis considers only the semi-active adversarial scenario, which requires a trusted source.

Our anonymous quantum communication protocol opens the way to the integration and implementation of this fundamental functionality into quantum networks currently under development.

{\bf Acknowledgments}.--- We acknowledge support of the European Union's Horizon 2020 Research and Innovation Programme under Grant Agreement No. 820445 (QIA), the ANR through the ANR-17-CE24-0035 VanQuTe and ANR-17-CE39-0005 quBIC projects, the BPI France project RISQ, the EPSRC (UK), and the MIT-France International Science and Technology Initiative.

\section{Appendix A: anonymous classical protocols}
We first give the \textsf{Parity} protocol from \cite{BT:asiacrypt07}, by which a set of $n$ agents can privately determine the parity of their input string (or equivalently, the XOR of their input bits), in Protocol \ref{alg:parity}. Note that although this uses a simultaneous broadcast channel, we only use the modified version of this protocol (as given in the \textsf{LogicalOR} protocol afterwards), which just requires a regular broadcast channel.

\begin{algorithm}[h]
\caption{\textsf{Parity} \cite{BT:asiacrypt07}} 
\label{alg:parity}
\begin{flushleft}
\textit{Input}: $\{ x_i \}_{i=1}^n$. \\
\textit{Goal}: Each agent gets $y_i = \bigoplus_{i=1}^n x_i$.
\end{flushleft}
\begin{algorithmic}[1]
\STATE Each of the $n$ agents wants to input their bit $x_i$. Every agent $i$ chooses random bits $\{r_i^j \}_{j=1}^n$ such that $\bigoplus_{j=1}^n r_i^j = x_i$. \\ \
\STATE Every agent $i$ sends their $j$th bit $r_i^j$ to agent $j$ ($j$ can equal $i$). \\ \
\STATE Every agent $j$ computes $z_j=\bigoplus_{i=1}^n r_i^j$ and reports the value in the simultaneous broadcast channel. \\ \
\STATE The value $z=\bigoplus_{j=1}^n z_j$ is computed, which equals $y_i$.
\end{algorithmic}
\end{algorithm}

This protocol is then used to construct the \textsf{LogicalOR} protocol \cite{BT:asiacrypt07}, by which a set of $n$ agents can privately determine the logical OR of their inputs. Due to repeating the protocol with different orderings of the agents each time, the simultaneous broadcast channel is no longer required. This is given in Protocol \ref{alg:logicalor}.

\begin{algorithm}[h]
\caption{\textsf{LogicalOR} \cite{BT:asiacrypt07}} 
\label{alg:logicalor}
\begin{flushleft}
\textit{Input}: $\{ x_i \}_{i=1}^n$, security parameter $S$. \\
\textit{Goal}: Each agent gets $y_i = \bigvee_{i=1}^n x_i$.
\end{flushleft}
\begin{algorithmic}[1]
\STATE The agents agree on $n$ orderings, with each ordering having a different last participant. \\ \
\STATE For each ordering:
\begin{enumerate}
\item[(a)] Each agent $i$ picks the value of $p_i$ as follows: if $x_i=0$, then $p_i=0$; if $x_i=1$, then $p_i=1$ with probability $\frac{1}{2}$ and $p_i=0$ with probability $\frac{1}{2}$. 
\item[(b)] Run the \textsf{Parity} protocol with input $\{p_i\}_{i=1}^n$, with a regular broadcast channel rather than simultaneous broadcast, and with the agents broadcasting according to the current ordering. If the result is $1$, then $y_i = 1$. 
\item[(c)] Repeat steps 2(a) - 2(b) $S$ times in total. If the result of the \textsf{Parity} protocol is never $1$, then $y_i = 0$.
\end{enumerate}
\end{algorithmic}
\end{algorithm}

If all agents input $x_i = 0$, then the \textsf{LogicalOR} protocol is correct with probability 1, however if any agent inputs $x_i = 1$, then the correctness is $(1 - 2^{-S})$.
In Protocol 5, the functionality \textsf{RandomBit} is used to pick between the verification and use of the state. This in turn calls \textsf{LogicalOR}. In an honest run, the outcome of \textsf{RandomBit} is the outcome of the Sender, and so if any agent behaves dishonestly the Sender will abort.
The functionality \textsf{RandomAgent}, which also calls \textsf{LogicalOR}, is simply a repetition of \textsf{RandomBit}, and so the same argument holds here.

\section{Appendix B: Proof of Theorem \ref{th:prob}}
Here, we prove the soundness of the protocol.
\begin{manualtheorem}{1}
Let $C_\epsilon$ be the event that the protocol does not abort and the state used for the anonymous transmission is such that  $F'(\ket{\Psi}) = \sqrt{1-\epsilon^2}$. Then for the honest agents, for all $\epsilon >0$,
\begin{equation}
\text{Pr}[C_\epsilon] \leq  2^{-S}\frac{4n}{1 - \sqrt{1 - \epsilon^2}}.
\end{equation}\label{th:prob}
\end{manualtheorem}
\begin{proof}
Our aim is to bound the probability that the protocol does not abort and the fidelity of the state $\ket{\Psi}$ used for anonymous transmission is given by $F'(\ket{\Psi}) = \underset{U}{\max \ } F(U \ket{\Psi}, \ket{\Phi_0^n}) = \sqrt{1 - \epsilon^2}$, where $U$ is a general operator on the space of the malicious agents.

Although we allow the malicious source to create any state in any round and even entangle the states between rounds, the optimal cheating strategy, which maximises the probability of the event $C_\epsilon$, is to create in each round some pure state $\ket{\Psi}$ such that $F'(\ket{\Psi}) = \sqrt{1-\epsilon^2}$, as proved in \cite{PAW:prl12}. In high level, one can first see that an entangled strategy does not help, as it can be replaced by a strategy sending unentangled states as follows. Given some entangled state, for a given round, the probability of passing the test and the fidelity of the state depend only on the reduced state, conditioned on passing previous rounds. The exact same effect can be achieved by sending these mixed reduced states corresponding to each round, without any entanglement.

Next, one sees that by providing a mixed state, the source does not gain any advantage, as a mixed state is a probabilistic mixture of pure states, and the overall cheating probability of this mixed strategy is just a weighted combination of the cheating probabilities of each of the pure states. Then, obviously this mixed strategy is worse than the strategy that always sends the pure state that has the maximum cheating probability of all states in the mixture. Hence, one can continue the proof by only considering strategies with pure states.  

Moreover, since the adversary is just trying to maximise the probability the state $\ket{\Psi}$ used for anonymous transmission has $F'(\ket{\Psi})= \sqrt{1 - \epsilon^2}$, it is clear that there is no need to send any state with even smaller $F'(\ket{\Psi})$, since then the probability of failing the test (and therefore the protocol aborting) would just increase. Last, if in any round the source created a state with higher $F'(\ket{\Psi})$, then this certainly does not contribute to the event $C_\epsilon$, and in fact it may also cause the protocol to abort. Thus, to upper-bound the probability of event $C_\epsilon$ with respect to the best attack a malicious source can perform, we only need to consider the case where in each round the malicious source creates some state $\ket{\Psi}$ such that $F'(\ket{\Psi}) = \sqrt{1-\epsilon^2}$.

First, we consider the probability that the state is used in round $l$. For this to happen, the Sender must get the result of all $S$ coin flips to be heads ($x=0$), which happens with probability $2^{-S}$. The Sender then calls \textsf{RandomBit} with her input as 0. The output of \textsf{RandomBit} will then be 0 with probability 1, since all agents input 0 (note that the Sender can see if any agent behaves dishonestly and inputs 1, as the output of \textsf{RandomBit} can then be 1, and so the Sender will abort).

Second, we consider the probability that the state is tested in all $(l-1)$ previous rounds. Here, the Sender must input 1 to \textsf{RandomBit}. First, the probability of the Sender not getting all $S$ coin flips to be 0 is given by $1-2^{-S}$. Then, the probability that \textsf{RandomBit} will give an output of 1 is given by $1 - 2^{-S}$ after $S$ rounds. Malicious agents will not affect the output, since even if they input 1 to \textsf{RandomBit}, the output will still be 1. Thus, the overall probability is given by $[(1-2^{-S})(1-2^{-S})]^{l-1}$.

Finally, we consider the probability that all the $(l-1)$ tests have passed. In our protocol, a randomly chosen agent $j$ runs the \textsf{Verification} protocol as the Verifier. If the Verifier is honest (which happens with probability $\frac{k}{n}$), the probability that the test is passed with a state $\ket{\Psi}$ is given by $P(\ket{\Psi})$. If the Verifier is malicious (with probability $ \frac{n-k}{n}$), we take the probability to be 1 as the worst case scenario. Then, we can write the probability that all $(l-1)$ tests have passed as $
\big( \frac{n-k}{n} + \frac{k}{n} P(\ket{\Psi}) \big)^{l-1}$. Note that from \cite{PAW:prl12}, the probability that a state $\ket{\Psi}$ with fidelity $F'(\ket{\Psi})$ will pass the test is given by $P(\ket{\Psi}) \leq \frac{3}{4} + \frac{F'}{4}$.

Thus, the total probability of event $C_{\epsilon}$ at the $l^{th}$ repetition of the protocol is:
\begin{align}
Pr [C_{\epsilon}^l] & \leq 2^{-S} \Big( 1 - 2^{1-S}  + 2^{-2S} \Big)^{l-1} \Big( 1 - \big( \frac{k - F'k}{4n} \big) \Big)^{l-1}.
\end{align}
We then take the integral to upper bound this probability as follows:
\begin{align}
Pr[C_{\epsilon}] & \leq \int_{0}^{\infty} 2^{-S} (1 - 2^{1-S} + 2^{-2S} )^l \Big( 1 - \big( \frac{k - F'k}{4n} \big)\Big)^{l} dl \\
& \leq 2^{-S}  \int_{0}^{\infty} \Big( 1 - \big( \frac{k - F'k}{4n} \big)\Big)^{l} dl \\
& = - \frac{2^{-S}}{\log{(1 - \big( \frac{k - F'k}{4n} \big) )}} \\
 & \leq 2^{-S} \frac{4n}{k (1 - F')} \\
&  \leq 2^{-S} \frac{4n}{k (1 - \sqrt{1 - \epsilon^2})}.
\end{align}
Since each honest agent does not know which other agents are honest or malicious, we can further upper-bound this in terms of a security statement for the honest agents:
\begin{align}
\text{Pr}[C_\epsilon] \leq 2^{-S} \frac{4n}{1-\sqrt{1-\epsilon^2}}.
\end{align}
If the agents take $S = \log_2 (\frac{4n}{(1-\sqrt{1-\epsilon^2}) \delta})$, they get $\text{Pr}[C_\epsilon] \leq \delta$. The expected number of runs of the protocol is given by $2^S = \frac{4n}{(1-\sqrt{1-\epsilon^2}) \delta}$. Thus, they can make this probability of failure negligible by doing a large number of runs.
\end{proof}

\section{Appendix C: Proof of Theorem \ref{th:anon}}
Next, we prove the anonymity of the protocol. For simplicity of the proof, recall that we denote the ideal state by $\ket{\Phi_0^n}$,
which can be obtained from the GHZ state by applying a Hadamard and a phase shift $\sqrt{Z}$ to each qubit. The Sender's transformation now becomes $\sigma_x \sigma_z$.
Further, we also define the state:
\begin{align}
 \ket{\Phi_1^n}  = \frac{1}{\sqrt{2^{n-1}}} \Big[ \underset{\Delta(y) = 1 \text{ (mod 4)}}{\sum} \ket{y} - \underset{\Delta(y) = 3 \text{ (mod 4)}}{\sum} \ket{y} \Big],
 \end{align}
and note that $\sigma_x \sigma_z \ket{\Phi_0^n} = \ket{\Phi_1^n}, \sigma_x \sigma_z \ket{\Phi_1^n} = - \ket{\Phi_0^n}$.

We consider two cases here: first, when all the agents are honest (Lemma \ref{l:honest}), and secondly, when we have malicious agents who could apply some operation on their part of the state (Lemma \ref{l:dishonest}).
\renewcommand{\thetheorem}{2A}
\begin{lemma}
If all the agents are honest, and they share a state $\ket{\Psi}$ such that $F(\ket{\Psi}, \ket{\Phi_0^n}) = \sqrt{1 - \epsilon^2}$, then for every honest agent $i, j$ who could be the Sender, we have that $F(\ket{\Psi_i}, \ket{\Psi_j}) \geq 1- \epsilon^2$, where $\ket{\Psi_i}$ is the state after agent $i$ has applied the Sender's transformation.
\label{l:honest}
\end{lemma}
\begin{proof}
If we have $F(\ket{\Psi}, \ket{\Phi_0^n}) = \abs{\bra{\Psi}\ket{\Phi_0^n}}^2 = \sqrt{1 - \epsilon^2}$, then similarly to \cite{PAW:prl12} we can write the state shared by all the agents as:
\begin{align}
\ket{\Psi} = (1 - \epsilon^2)^{1/4} \ket{\Phi_0^n} + \epsilon_1 \ket{\Phi_1^n} + \sum_{i=2}^{2^n-1} \epsilon_i \ket{\Phi_i^n},
\end{align}
where $\sum_{i=1}^{2^n-1} \epsilon_i^2 = 1 - \sqrt{1-\epsilon^2}$. If agent $i$ is the Sender, then she applies $\sigma_x \sigma_z$, and the state becomes:
\begin{align}
\ket{\Psi_i} = (1 - \epsilon^2)^{1/4} \ket{\Phi_1^n} - \epsilon_1 \ket{\Phi_0^n} + \sum_{i=2}^{2^n-1} \epsilon_i' \ket{\Phi_i^n}.
\end{align}
Instead, if agent $j$ is the Sender and she applies $\sigma_x \sigma_z$, the state becomes:
\begin{align}
\ket{\Psi_j} = (1 - \epsilon^2)^{1/4} \ket{\Phi_1^n} - \epsilon_1 \ket{\Phi_0^n} + \sum_{i=2}^{2^n-1} \epsilon_i'' \ket{\Phi_i^n}.
\end{align}
The fidelity is then given by:
\begin{align}
F(\ket{\Psi_i}, \ket{\Psi_j}) & = \abs{\bra{\Psi_i}\ket{\Psi_j}}^2 \\
& = \abs{\sqrt{1 - \epsilon^2} + \epsilon_1^2 + \sum_{i=2}^{2^n-1} \epsilon_i' \epsilon_i''}^2 \\
& \geq 1 - \epsilon^2.
\end{align}
\end{proof}

\renewcommand{\thetheorem}{2B}
\begin{lemma}
If some of the agents are malicious, and they share a state $\ket{\Psi}$ such that $F'(\ket{\Psi}) \geq \sqrt{1-\epsilon^2}$, then for every honest agent $i, j$ who could be the Sender, we have that $F(\ket{\Psi_i}, \ket{\Psi_j}) \geq 1 - \epsilon^2 $, where $\ket{\Psi_i}$ is the state after agent $i$ has applied the Sender's transformation.
\label{l:dishonest}
\end{lemma}
\begin{proof}
Recall that our fidelity measure is given by $F'(\ket{\Psi}) = \underset{U}{\max \ } F(U\ket{\Psi}, \ket{\Phi_0^n})$. Let us now denote by $\ket{\Psi'}=U \ket{\Psi}$ the state after the operation $U$ which maximises this fidelity has been applied. As in \cite{PAW:prl12}, we can write this state in the most general form as:
\begin{align}
\ket{\Psi'} = \ket{\Phi_0^k} \ket{\psi_0} + \ket{\Phi_1^k} \ket{\psi_1} + \ket{\chi},
\end{align}
where note that $\ket{\chi}$ contains both honest and malicious parts, of which the honest part is orthogonal to both $\ket{\Phi_0^k}$ and $\ket{\Phi_1^k}$.

We want to find the closeness of the states $\ket{\Psi_i}, \ket{\Psi_j}$, which are the states after the $\sigma_x \sigma_z$ operation is applied to $\ket{\Psi'}$ by either agent $i$ or $j$ who is the Sender. 
These states are given by:
\begin{align}
\ket{\Psi_i} & = \ket{\Phi_1^k} \ket{\psi_0} - \ket{\Phi_0^k} \ket{\psi_1} + \ket{\chi'}, \\
\ket{\Psi_j} & = \ket{\Phi_1^k} \ket{\psi_0} - \ket{\Phi_0^k} \ket{\psi_1} + \ket{\chi''}.
\end{align}
The fidelity is then given by:
\begin{align}
F(\ket{\Psi_i}, \ket{\Psi_j}) & = \abs{\bra{\Psi_i}\ket{\Psi_j}}^2 \\
& = \abs{\bra{\psi_0}\ket{\psi_0} + \bra{\psi_1}\ket{\psi_1} + \bra{\chi'}\ket{\chi''}}^2.
\end{align}
However, although the overall state $\ket{\Psi'}$ is normalised, the malicious agents' part of the state is not. Thus, we need to determine a bound on $\bra{\psi_0}\ket{\psi_0}$ and $\bra{\psi_1}\ket{\psi_1}$. We have: 
\begin{align}
F( \ket{\Psi'}, \ket{\Phi_0^n}) =  \abs{\bra{\Phi_0^n}\ket{\Psi'}}^2  \geq \sqrt{1 - \epsilon^2}.
\end{align}
It was shown in \cite{PAW:prl12} that we can write for any $k, n$:
\begin{align}
\ket{\Phi_0^n} = \frac{1}{\sqrt{2}} \Big[ \ket{\Phi_0^k} \ket{\Phi_0^{n-k}} - \ket{\Phi_1^k} \ket{\Phi_1^{n-k}} \Big],
\end{align}
and using this, we get:
 \begin{align}
 \frac{1}{2} | & (\bra{\Phi_0^{n-k}}\ket{\psi_0})^2 + (\bra{\Phi_1^{n-k}}\ket{\psi_1})^2 \nonumber \\
 & - 2 \bra{\Phi_0^{n-k}}\ket{\psi_0} \bra{\Phi_1^{n-k}}\ket{\psi_1} |  \geq \sqrt{1 - \epsilon^2}.
 \end{align}
Using the triangle inequality, we have:
 \begin{align}
  \frac{1}{2} \Big[ \abs{\bra{\Phi_0^{n-k}}\ket{\psi_0}}^2 & + \abs{\bra{\Phi_1^{n-k}}\ket{\psi_1}}^2  \Big]  \geq \sqrt{1 - \epsilon^2 }.
\end{align}
Using the Cauchy-Schwarz inequality, we have:
\begin{align}
\bra{\psi_0}\ket{\psi_0} + \bra{\psi_1}\ket{\psi_1} & \geq \abs{\bra{\Phi_0^{n-k}}\ket{\psi_0}}^2  + \abs{\bra{\Phi_1^{n-k}}\ket{\psi_1}}^2 \\
& \geq \sqrt{1 - \epsilon^2}.
\end{align}
Since the overall state $\ket{\Psi'}$ is normalised, we have $\bra{\chi'}\ket{\chi''} \leq 1 - \sqrt{1 - \epsilon^2}$. Thus, we get our expression for fidelity as:
 \begin{align}
 F(\ket{\Psi_i}, \ket{\Psi_j}) & =  \abs{\bra{\psi_0}\ket{\psi_0} + \bra{\psi_1}\ket{\psi_1} + \bra{\chi'}\ket{\chi''}}^2 \\
 & \geq 1 - \epsilon^2.
 \end{align}
\end{proof}

We are now ready to prove Theorem \ref{th:anon}.
\begin{manualtheorem}{2}
If the agents share a state $\ket{\Psi}$ such that $F'(\ket{\Psi}) \geq \sqrt{1-\epsilon^2}$, then the probability that the malicious agents can guess the identity of the Sender is given by:
\begin{align}
\text{Pr}[\text{guess}] & \leq \frac{1}{k} + \epsilon.
\end{align}
\label{th:anon}
\end{manualtheorem}
\begin{proof}
We will now show that if the agents share close to the GHZ state, then the Sender remains anonymous.
From Theorem \ref{th:prob}, we saw that the probability that the state used for anonymous transmission satisfies $F'(\ket{\Psi}) \leq \sqrt{1 - \epsilon^2}$ is given by $
\text{Pr}[C_\epsilon] \leq \delta$
for the honest agents, where $\delta$ depends on the number of runs of the verification protocol. Thus, by doing enough runs, we can make this very small, and so we have that the state used for anonymous transmission will be close to the GHZ state, as given by $F'(\ket{\Psi}) \geq \sqrt{1 - \epsilon^2}$.

From the previous proof, we see that if $F'(\ket{\Psi}) \geq \sqrt{1-\epsilon^2}$, the distance between the states if agent $i$ or $j$ was the Sender is $D(\ket{\Psi_i}, \ket{\Psi_j}) \leq \epsilon$. A malicious agent who wishes to guess the identity of the Sender would make some sort of measurement to do so. Thus, we wish to find the maximum success probability of a measurement that could distinguish between the $k$ states that are the result of the Sender (who can only be an honest agent) applying the $\sigma_x \sigma_z$ transformation.

The success probability of discriminating between $k$ states is given by $ \sum_{i=1}^k p_i \text{Tr} (\Pi_i \rho_i)$. From Lemma \ref{l:dishonest}, we know that the distance between any two states after the Sender's transformation is upper-bounded by $\epsilon$. Thus, if we take $\ket{\alpha} = \ket{\Psi_j}$, then we know that any of these $k$ states is of distance $\epsilon$ away from this same state $\ket{\alpha}$.

For any POVM element $P$, we can write the trace distance between two states $\rho, \sigma$ as $
 \text{Tr} \big[ P (\rho - \sigma) \big] \leq D(\rho, \sigma)$.
Thus, we have for a POVM element $\Pi_i$ and for states $\ket{\Psi_i}, \ket{\alpha}$:
\begin{align}
 \text{Tr} (\Pi_i \ket{\Psi_i}\bra{\Psi_i}) - \text{Tr} (\Pi_i \ket{\alpha}\bra{\alpha})  \leq \epsilon.
\end{align}
Assuming that each honest agent has an equiprobable chance of becoming the Sender, the probability that the malicious agents can guess the identity of the Sender is bounded by:
\begin{align}
\text{Pr}[\text{guess}]
& = \sum_{i=1}^k \frac{1}{k} \text{Tr} (\Pi_i \ket{\Psi_i}\bra{\Psi_i})
\\
& \leq \frac{1}{k} \sum_{i=1}^k \Big[ \text{Tr} (\Pi_i \ket{\alpha}\bra{\alpha}) + \epsilon
 \Big]  \\
& = \frac{1}{k} \text{Tr}\Big[\sum_{i=1}^k \Pi_i \ket{\alpha}\bra{\alpha}\Big] + \frac{1}{k} k  \epsilon  \\
& = \frac{1}{k} \text{Tr} (\ket{\alpha}\bra{\alpha}) + \epsilon \\
& = \frac{1}{k} + \epsilon.
\end{align}
\end{proof}


\begin{thebibliography}{12}
\expandafter\ifx\csname natexlab\endcsname\relax\def\natexlab#1{#1}\fi
\expandafter\ifx\csname bibnamefont\endcsname\relax
  \def\bibnamefont#1{#1}\fi
\expandafter\ifx\csname bibfnamefont\endcsname\relax
  \def\bibfnamefont#1{#1}\fi
\expandafter\ifx\csname citenamefont\endcsname\relax
  \def\citenamefont#1{#1}\fi
\expandafter\ifx\csname url\endcsname\relax
  \def\url#1{\texttt{#1}}\fi
\expandafter\ifx\csname urlprefix\endcsname\relax\def\urlprefix{URL }\fi
\providecommand{\bibinfo}[2]{#2}
\providecommand{\eprint}[2][]{\url{#2}}

\bibitem[{\citenamefont{Kimble}(2008)}]{Kim:nature08}
\bibinfo{author}{\bibfnamefont{J.}~\bibnamefont{Kimble}},
  \bibinfo{journal}{Nature} \textbf{\bibinfo{volume}{453}},
  \bibinfo{pages}{1023} (\bibinfo{year}{2008}).

\bibitem[{\citenamefont{Scarani et~al.}(2009)\citenamefont{Scarani,
  Bechmann-Pasquinucci, Cerf, Dusek, L\"utkenhaus, and Peev}}]{SBC:rmp09}
\bibinfo{author}{\bibfnamefont{V.}~\bibnamefont{Scarani}},
  \bibinfo{author}{\bibfnamefont{H.}~\bibnamefont{Bechmann-Pasquinucci}},
  \bibinfo{author}{\bibfnamefont{N.~J.} \bibnamefont{Cerf}},
  \bibinfo{author}{\bibfnamefont{M.}~\bibnamefont{Dusek}},
  \bibinfo{author}{\bibfnamefont{N.}~\bibnamefont{L\"utkenhaus}},
  \bibnamefont{and} \bibinfo{author}{\bibfnamefont{M.}~\bibnamefont{Peev}},
  \bibinfo{journal}{Rev. Mod. Phys.} \textbf{\bibinfo{volume}{81}},
  \bibinfo{pages}{1301} (\bibinfo{year}{2009}).

\bibitem[{\citenamefont{Diamanti et~al.}(2016)\citenamefont{Diamanti, Lo, Qi,
  and Yuan}}]{DLQ:npjqi16}
\bibinfo{author}{\bibfnamefont{E.}~\bibnamefont{Diamanti}},
  \bibinfo{author}{\bibfnamefont{H.-K.} \bibnamefont{Lo}},
  \bibinfo{author}{\bibfnamefont{B.}~\bibnamefont{Qi}}, \bibnamefont{and}
  \bibinfo{author}{\bibfnamefont{Z.}~\bibnamefont{Yuan}}, \bibinfo{journal}{npj
  Quantum Info.} \textbf{\bibinfo{volume}{2}}, \bibinfo{pages}{16025}
  (\bibinfo{year}{2016}).

\bibitem[{\citenamefont{Gheorghiu et~al.}(2018)\citenamefont{Gheorghiu,
  Kapourniotis, and Kashefi}}]{GKK:tcs18}
\bibinfo{author}{\bibfnamefont{A.}~\bibnamefont{Gheorghiu}},
  \bibinfo{author}{\bibfnamefont{T.}~\bibnamefont{Kapourniotis}},
  \bibnamefont{and} \bibinfo{author}{\bibfnamefont{E.}~\bibnamefont{Kashefi}},
  \bibinfo{journal}{Theory Comput. Syst.}  (\bibinfo{year}{2018}),
  \bibinfo{note}{https://doi.org/10.1007/s00224-018-9872-3}.

\bibitem[{\citenamefont{Broadbent and Tapp}(2007)}]{BT:asiacrypt07}
\bibinfo{author}{\bibfnamefont{A.}~\bibnamefont{Broadbent}} \bibnamefont{and}
  \bibinfo{author}{\bibfnamefont{A.}~\bibnamefont{Tapp}}, in
  \emph{\bibinfo{booktitle}{Proc. ASIACRYPT}} (\bibinfo{year}{2007}), pp.
  \bibinfo{pages}{410--426}.

\bibitem[{\citenamefont{Christandl and Wehner}(2005)}]{CW:asiacrypt05}
\bibinfo{author}{\bibfnamefont{M.}~\bibnamefont{Christandl}} \bibnamefont{and}
  \bibinfo{author}{\bibfnamefont{S.}~\bibnamefont{Wehner}}, in
  \emph{\bibinfo{booktitle}{Proc. ASIACRYPT}} (\bibinfo{year}{2005}), pp.
  \bibinfo{pages}{217--235}.

\bibitem[{\citenamefont{Greenberger et~al.}(1989)\citenamefont{Greenberger,
  Horne, and Zeilinger}}]{GHZ}
\bibinfo{author}{\bibfnamefont{D.~M.} \bibnamefont{Greenberger}},
  \bibinfo{author}{\bibfnamefont{M.~A.} \bibnamefont{Horne}}, \bibnamefont{and}
  \bibinfo{author}{\bibfnamefont{A.}~\bibnamefont{Zeilinger}}, in
  \emph{\bibinfo{booktitle}{Bell's Theorem, Quantum Theory, and Conceptions of
  the Universe, M. Kafatos (Ed.), Kluwer, Dordrecht}} (\bibinfo{year}{1989}),
  pp. \bibinfo{pages}{69--72}.

\bibitem[{\citenamefont{Bennett et~al.}(1993)\citenamefont{Bennett, Brassard,
  Cr\'epeau, Jozsa, Peres, and Wootters}}]{BBC:prl93}
\bibinfo{author}{\bibfnamefont{C.~H.} \bibnamefont{Bennett}},
  \bibinfo{author}{\bibfnamefont{G.}~\bibnamefont{Brassard}},
  \bibinfo{author}{\bibfnamefont{C.}~\bibnamefont{Cr\'epeau}},
  \bibinfo{author}{\bibfnamefont{R.}~\bibnamefont{Jozsa}},
  \bibinfo{author}{\bibfnamefont{A.}~\bibnamefont{Peres}}, \bibnamefont{and}
  \bibinfo{author}{\bibfnamefont{W.~K.} \bibnamefont{Wootters}},
  \bibinfo{journal}{Phys. Rev. Lett.} \textbf{\bibinfo{volume}{70}},
  \bibinfo{pages}{1895} (\bibinfo{year}{1993}).

\bibitem[{\citenamefont{Lipinska et~al.}(2018)\citenamefont{Lipinska, Murta,
  and Wehner}}]{LMW:pra18}
\bibinfo{author}{\bibfnamefont{V.}~\bibnamefont{Lipinska}},
  \bibinfo{author}{\bibfnamefont{G.}~\bibnamefont{Murta}}, \bibnamefont{and}
  \bibinfo{author}{\bibfnamefont{S.}~\bibnamefont{Wehner}},
  \bibinfo{journal}{Phys. Rev. A} \textbf{\bibinfo{volume}{98}},
  \bibinfo{pages}{052320} (\bibinfo{year}{2018}).

\bibitem[{\citenamefont{Brassard et~al.}(2007)\citenamefont{Brassard,
  Broadbent, Fitzsimons, Gambs, and Tapp}}]{BBF:asiacrypt07}
\bibinfo{author}{\bibfnamefont{G.}~\bibnamefont{Brassard}},
  \bibinfo{author}{\bibfnamefont{A.}~\bibnamefont{Broadbent}},
  \bibinfo{author}{\bibfnamefont{J.}~\bibnamefont{Fitzsimons}},
  \bibinfo{author}{\bibfnamefont{S.}~\bibnamefont{Gambs}}, \bibnamefont{and}
  \bibinfo{author}{\bibfnamefont{A.}~\bibnamefont{Tapp}}, in
  \emph{\bibinfo{booktitle}{Proc. ASIACRYPT}} (\bibinfo{year}{2007}), pp.
  \bibinfo{pages}{460--473}.

\bibitem[{\citenamefont{Pappa et~al.}(2012)\citenamefont{Pappa, Chailloux,
  Wehner, Diamanti, and Kerenidis}}]{PAW:prl12}
\bibinfo{author}{\bibfnamefont{A.}~\bibnamefont{Pappa}},
  \bibinfo{author}{\bibfnamefont{A.}~\bibnamefont{Chailloux}},
  \bibinfo{author}{\bibfnamefont{S.}~\bibnamefont{Wehner}},
  \bibinfo{author}{\bibfnamefont{E.}~\bibnamefont{Diamanti}}, \bibnamefont{and}
  \bibinfo{author}{\bibfnamefont{I.}~\bibnamefont{Kerenidis}},
  \bibinfo{journal}{Phys. Rev. Lett.} \textbf{\bibinfo{volume}{108}},
  \bibinfo{pages}{260502} (\bibinfo{year}{2012}).

\bibitem[{\citenamefont{McCutcheon et~al.}(2016)\citenamefont{McCutcheon,
  Pappa, Bell, McMillan, Chailloux, Lawson, Mafu, Markham, Diamanti, Kerenidis
  et~al.}}]{MPB:natcomm16}
\bibinfo{author}{\bibfnamefont{W.}~\bibnamefont{McCutcheon}},
  \bibinfo{author}{\bibfnamefont{A.}~\bibnamefont{Pappa}},
  \bibinfo{author}{\bibfnamefont{B.~A.} \bibnamefont{Bell}},
  \bibinfo{author}{\bibfnamefont{A.}~\bibnamefont{McMillan}},
  \bibinfo{author}{\bibfnamefont{A.}~\bibnamefont{Chailloux}},
  \bibinfo{author}{\bibfnamefont{T.}~\bibnamefont{Lawson}},
  \bibinfo{author}{\bibfnamefont{M.}~\bibnamefont{Mafu}},
  \bibinfo{author}{\bibfnamefont{D.}~\bibnamefont{Markham}},
  \bibinfo{author}{\bibfnamefont{E.}~\bibnamefont{Diamanti}},
  \bibinfo{author}{\bibfnamefont{I.}~\bibnamefont{Kerenidis}},
  \bibnamefont{et~al.}, \bibinfo{journal}{Nature Commun.}
  \textbf{\bibinfo{volume}{7}}, \bibinfo{pages}{13251} (\bibinfo{year}{2016}).

\end{thebibliography}
\end{document}